\documentclass[runningheads]{llncs}

\usepackage[english]{babel}

\usepackage{microtype}
\usepackage{fullpage}
\usepackage{amssymb}
\usepackage{amsmath}
\usepackage{graphicx}
\usepackage[colorlinks=true, allcolors=blue]{hyperref}
\usepackage{mathtools}
\usepackage{tikz}
\usepackage{algorithm}
\usepackage{algorithmic}
\usepackage{subcaption}
\usepackage{complexity}
\usepackage{enumitem}

\newcommand{\B}{\mathbb{B}}

\newcommand{\post}{\mathrm{post}}
\newcommand{\paths}{\mathrm{paths}}

\newcommand{\prop}{\mathcal{P}}
\newcommand{\subf}{\mathit{sub}}
\newcommand{\propformula}{\Omega}
\DeclareMathOperator{\lF}{\mathbf{F}}
\DeclareMathOperator{\lG}{\mathbf{G}}
\DeclareMathOperator{\lU}{\mathbf{U}}
\DeclareMathOperator{\lX}{\mathbf{X}}

\DeclareMathOperator{\false}{\mathit{false}}
\DeclareMathOperator{\true}{\mathit{true}}
\DeclareMathOperator{\lA}{\mathbf{A}}
\DeclareMathOperator{\lE}{\mathbf{E}}

\newcommand{\placeholder}{?}

\newcommand{\infer}{\mathtt{Infer}}
\newcommand{\Sat}{\mathtt{SAT}_{M}}
\newcommand{\sample}{\mathcal{S}}

\begin{document}
\title{Inferring Properties in Computation Tree Logic}
\author{Rajarshi Roy\inst{1} \and Daniel Neider\inst{2,3}}

\institute{Max Planck Institute for Software Systems, Kaiserslautern, Germany \and TU Dortmund University, Dortmund, Germany 
\and Center for Trustworthy Data Science and Security, University Alliance Ruhr, Dortmund, Germany }

\maketitle

\begin{abstract}
We consider the problem of automatically inferring specifications in the branching-time logic, Computation Tree Logic (CTL), from a given system.
Designing functional and usable specifications has always been one of the biggest challenges of formal methods.
While in recent years, works have focused on automatically designing specifications in linear-time logics such as Linear Temporal Logic (LTL) and Signal Temporal Logic (STL), little attention has been given to branching-time logics despite its popularity in formal methods.
We intend to infer concise (thus, interpretable) CTL formulas from a given finite state model of the system in consideration.
However, inferring specification only from the given model (and, in general, from only positive examples) is an ill-posed problem.
As a result, we infer a CTL formula that, along with being concise, is also language-minimal, meaning that it is rather specific to the given model.
We design a counter-example guided algorithm to infer a concise and language-minimal CTL formula via the generation of undesirable models.
In the process, we also develop, for the first time, a passive learning algorithm to infer CTL formulas from a set of desirable and undesirable Kripke structures.
The passive learning algorithm involves encoding a popular CTL model-checking procedure in the Boolean Satisfiability problem.
\end{abstract}

\section{Introduction}
Formal verification relies on the fact that formal specifications that express the design requirements precisely are readily available or can be constructed easily.
However, this is often an unrealistic assumption since constructing specifications manually is a tedious and error-prone task.
For many years, the availability of functional and usable specifications has been widely believed to be one of the biggest bottlenecks of formal methods~\cite{AmmonsBL02,Rozier16,BjornerH14}. 

Consequently, there have been several works in automatically constructing specifications in a number of temporal logics.
Most of the notable works are targeted towards two popular temporal logics: Linear Temporal Logic (LTL)~\cite{flie,CamachoM19,scarlet} and Signal Temporal Logic (STL)~\cite{dtmethod,MohammadinejadD20}.

However, there is little work that focuses on branching time logics, such as Computation Tree Logic (CTL), although they have been extensively used in Formal Verification~\cite{CimattiCGR00,Biere09}.
Chan~\cite{Chan00} considered the problem of completing queries in CTL, which are incomplete CTL specifications with a missing proposition (or a Boolean combination of them), to infer the strongest invariants.
Wasylkowski and Zeller~\cite{WasylkowskiZ11} mine invariants in CTL for Java programs using queries of
the form $\lA \lF \placeholder$, $\lA \lG (\placeholder_1 \rightarrow \lF\placeholder_2)$.
However, the approaches are limited in their ability to design arbitrary CTL specifications for the system under consideration.

To alleviate the shortcomings in the existing approaches, we devise the first algorithm for inferring interpretable formulas in Computation Tree Logic (CTL) (allowing an arbitrary structure) from a given finite state representation of the system under consideration.
For inferring formulas in CTL, following the existing problem settings~\cite{Chan00,WasylkowskiZ11}, we choose \emph{Kripke structures} as the finite-state representation of systems.

As observed by Roy et al.~\cite{ltl-from-positive-only}, inferring 
temporal logic properties from only positive examples---in this case, just a Kripke structure---is typically an ill-posed problem.
This is because one can always infer trivial formulas such as $\true$ that will hold in any input model.
Thus, to infer meaningful formulas, we rely on two regularizers to search for CTL specifications: the size and the language ((i.e., the set of allowed models) of the specification as also done by Roy et al.
In particular, we solve the following problem:
given a model $M$ and a size bound $b$, find a CTL formula $\Phi$ such that (i) $\Phi$ holds in $M$ (ii) the size of $\Phi$ is smaller than $b$; and (iii) $\Phi$ is language-minimal.
In this problem, the size bound prohibits the inferred formula from explicitly describing the model verbosely. 
The language-minimality criterion, on the other hand, ensures that the inferred formula is not a small trivial formula such as $\true$.
We present more details on the problem formulation in Section~\ref{sec:problem-setting}.

To tackle the above problem, we design a counterexample-guided (CEG) algorithm inspired by Avellaneda and Petrenko~\cite{AvellanedaP18} and Roy et al.~\cite{ltl-from-positive-only}.
The CEG algorithm relies on generating suitable negative examples (in this case, undesirable Kripke structures) to iteratively guide the inference process.
We rely on checking the satisfiability of CTL formulas to generate such negative examples.
We expand on the CEG algorithm in Section~\ref{sec:ceg-algorithm}.

The CEG algorithm relies internally on a \emph{passive learning} algorithm for CTL, which, given a set of desirable Kripke structures and a set of undesirable Kripke structures, learns a CTL formula that holds on the desirable Kripke structures and not on the undesirable ones.
In this paper, we devise a constraint-based passive learning algorithm inspired by the SAT-based algorithm of~\cite{flie}.
The learning algorithm involves symbolically encoding the fixed-point based model checking algorithm~\cite{model-checking-book} for CTL.
To the best of our knowledge, this is the only known work that devises a passive learning algorithm for the full class of CTL.
We describe the algorithm in Section~\ref{sec:passive-learning}.

\subsubsection*{Related Work}
As alluded to in the introduction, there is little work related to inferring formulas in branching-time logics.

One notable related work is the one by Chan~\cite{Chan00}, where he considers the problem of completing CTL queries.
CTL queries are incomplete CTL formulas with a single placeholder indicating the missing atomic proposition (or Boolean combination of them).
While there might be several possible ways of completing the missing part, Chan searches for the strongest invariant in CTL for the given model due to its relevance in model understanding.
Contrary to our approach, Chan's approach can only infer atomic propositions (and their Boolean combinations) but cannot handle complex formulas involving temporal operators.

The other work on inferring CTL formulas is by Wasylkowski and Zeller~\cite{WasylkowskiZ11}, who infer operational preconditions for methods of Java Programs in CTL.
They first construct a so-called object usage model from Java methods, which are then transformed into several Kripke structures.
They then search for CTL formulas that hold in the Kripke structures by iteratively searching through all possible CTL formulas and model-checking them against the Kripke structures to single out the suitable ones.
As they observed, their method does not scale to CTL formulas of arbitrary syntactic structure, and thus, they also resort to handcrafted queries (or templates in their terminology).

For inferring LTL formulas, there are several works that rely on handcrafted queries.
The most notable ones include~\cite{LiDS11,ShahKSL18,KimMSAS19}.
Few others can infer formulas of arbitrary syntactic structure in LTL (or its important fragments), such as~\cite{flie,CamachoM19,Riener19,scarlet}.
For inferring STL formulas, most works focus on inferring formulas of particular syntactic structure~\cite{dtbombara,dtmethod} or solving the parameter search problem for STL~\cite{asarin,rpstl1,rpstl2}.
There are only a handful of works inferring STL formulas of arbitrary structure~\cite{MohammadinejadD20,genetic}.
Finally, Arif et al.~\cite{ArifLERCT20} can infer arbitrary specifications in Past-time LTL and Roy et al.~\cite{0002FN20} can infer arbitrary specifications in Property Specification Language (PSL), an extension of LTL allowing regular expressions.
None of these approaches can be directly applied to learn concise CTL formulas for a given Kripke structure.

In terms of the approach used, the SAT encoding in our passive learning algorithm is similar in spirit to the SMT encoding by Bertrand et al.~\cite{BertrandFS12} for Probabilistic CTL (PCTL).
This encoding is slightly different from ours and was developed independently.
Also, it is designed for solving a different problem of synthesizing small models for a given PCTL formula.

\section{Problem Setting}
\label{sec:problem-setting}

We begin by setting up the main aspects of the central problem of the paper.

\subsection{Kripke structure}

We represent systems in a finite manner using Kripke structures.
A \textit{Kripke structure} $M = (S,I,R,L)$ over a set $\mathcal{P}$  of propositions consists of a set of states $S$, a set of initial states $I\subseteq S$, a transition relation $R\subseteq S\times S$ such that for all states $s$ there exists a state $s'$ satisfying $(s,s')\in R$, and a labeling function $L\colon S\rightarrow 2^{\mathcal{P}}$.
We say that $M$ is finite if it has finitely many states. In that case, we define the size of $M$ as $\lvert S\rvert$.
The set~$\post(s) = \{s'\in S \mid (s,s') \in R\}$ contains all successors of $s \in S$.

A (infinite) \emph{path} of the Kripke structure $M$ is an infinite sequence of states $\pi = s_0s_1\dots$ such that $s_{i+1} \in \post(s_i)$ for each $i\geq 0$.
A \emph{path prefix} of the Kripke structure $M$ is a finite sequence of states $\pi = s_0s_1\dots s_t$ such that $s_{i+1} \in \post(s_i)$ for each $0\leq i\leq t$.
The length $|\pi|$ of a path prefix is the number of states in the prefix $\pi$.
Furthermore, for a path (or a path prefix) $\pi$ and $0\leq i$ ($0\leq i< |\pi|$, respectively), let $\pi[i]$ denote the $i$-th state of $\pi$, and $\pi[i:]$ denote the suffix of $\pi$ from index $i$ on.
For a state $s$, let $\paths(s)$ be the set of all (infinite) paths starting from $s$.

\subsection{Computation Tree Logic (CTL)}

Computation Tree Logic (CTL) is a branching-time logic that reasons about the temporal behavior of (non-deterministic) systems.
To introduce the syntax of CTL formulas, we must introduce two types of formulas, state and path formulas. 
Intuitively, state formulas express properties of states, whereas path formulas express temporal properties of paths. 
For ease of notation, we denote state formulas and path formulas with Greek capital letters and Greek lowercase letters, respectively.
CTL state formulas over $\prop$ are given by the grammar
\[\Phi \Coloneqq p \mid \neg \Phi \mid  \Phi \wedge \Phi  \mid \lE \varphi \mid \lA\varphi,\]
where $p\in \prop$, and $\varphi$ is a path formula.
We include Boolean constants such as $\true$ and $\false$ and other standard Boolean formulas such as $\Phi_1\wedge\Phi_2$ and $\Phi_1\rightarrow\Phi_2$.
Next, CTL path formulas are given by the grammar
\[\varphi \Coloneqq \lX \Phi \mid \Phi \lU \Phi,\]
where $\lX$ represents the neXt operator, and $\lU$ the Until operator.
As syntax sugar, we allow standard temporal operators $\lF$, the Finally operator and $\lG$, the Globally operator, which are defined in the usual manner: $\circ\lF\Phi\coloneqq \circ(\true \lU \Phi)$ and $\circ\lG\Phi\coloneqq \neg\circ (\lF \neg\Phi)$ for $\circ\in\{\lE,\lA\}$.
We define the size $|\Phi|$ of a CTL formula as the size of the syntax DAG of $\Phi$.

We interpret CTL formulas over Kripke structures using the standard definitions~\cite{model-checking-book}.
Given a state $s$ and state formula~$\Phi$, we define when $\Phi$ holds 
 in state $s$, denoted using $s\models \Phi$, inductively as follows:
\begin{align*}
     s \models p & \text{ if and only if } p \in L(s), \\
     s \models \neg \Phi & \text{ if and only if } s \not\models \Phi, \\
    s\models \Phi_1\wedge\Phi_2 & \text{ if and only if } s\models\Phi_1 \text{ and } s\models\Phi_2, \\
    s\models \lE\varphi & \text{ if and only if } \pi\models\varphi \text{ for some } \pi\in\paths(s) \\
    s\models \lA\varphi & \text{ if and only if } \pi\models\varphi \text{ for all } \pi\in\paths(s)
\end{align*}
Similarly, given a path $\pi$ and a path formula $\varphi$, we define when $\varphi$ holds in path $\pi$, also denoted using $\pi\models \varphi$, inductively as follows:
\begin{align*}
\pi\models \lX\Phi &\text{ if and only if } \pi[1] \models \Phi \\
\pi\models \Phi_1 \lU \Phi_2 & \text{ if and only if } \text{there exists some } j\ge 0 \text{ such that } \pi[j]\models \Phi_2 \\
& \hspace{2cm}\text{ and for all } 0\le k < j, \pi[k]\models\Phi_1 
\end{align*}
We now say that a CTL formula $\Phi$ holds on a Kripke structure $M$, denoted by $M\models \Phi$, if $s\models\Phi$ for all initial states $s\in I$ of $M$.

To define the size of a CTL formula $\Phi$, we rely on its set of subformulas, denoted by $\subf(\Phi)$. 
It is defined recursively as follows:
$\subf(p) = \{p\}$, $\subf(\neg\Phi) = \{\neg\Phi\}
\cup\subf(\Phi)$, $\subf(\Phi_1\wedge\Phi_2) = \{\Phi_1\wedge\Phi_2\}\cup \subf(\Phi_1)\cup\subf(\Phi_2)$, $\subf(\circ\lX\Phi) = \{\circ\lX\Phi\}\cup\subf(\Phi)$, $\subf(\circ(\Phi_1\lU\Phi_2)) = \{\circ(\Phi_1\lU\Phi_2)\}\cup\subf(\Phi_1)\cup\subf(\Phi_2),$ where $\circ\in\{\lE,\lA\}$.
We now define the size $|\Phi|$ of a CTL formula as the number of its subformulas $|\subf(\Phi)|$.

\subsection{Problem Formulation}

We now state the central problem of the paper. 
\begin{problem}[CTL Inference]\label{prob:occ-ctl}
	Given a Kripke structure $M$ and a size bound $B$, find a CTL formula $\Phi$ such that:
	\begin{enumerate}
		\item $\Phi$ holds on $M$
            \item $|\Phi|\le B$; and
		\item for all CTL formulas $\Phi'$ that hold on $M$ and $|\Phi'|\le B$, $\Phi'\not\rightarrow\Phi$ or $\Phi\rightarrow\Phi'$.
	\end{enumerate}
\end{problem}
In the above problem, Condition 1 ensures that the prospective CTL formula $\Phi$ holds in the given model $M$. 
Condition 2 ensures that $\Phi$ is not too large.
This condition serves two purposes: it ensures $\Phi$ remains interpretable and also that $\Phi$ does not overfit $M$.
Condition 3 enforces that $\Phi$ holds on a minimal set of models when compared to all formulas $\Phi'$ that hold on $M$ and have size $|\Phi'|\le B$.
This condition prevents $\Phi$ from being a trivial formula such as $\true$ that conveys no information about the underlying system.
Together the three conditions ensure we obtain an interpretable CTL specification that concisely represents the possible behavior of the given model.

\section{A Counterexample-Guided (CEG) Algorithm}
\label{sec:ceg-algorithm}
The high-level design of our algorithm is to iteratively accumulate "negative" examples that direct the search towards the prospective CTL formula.
In particular, the algorithm generates a set $N$ of undesirable Kripke structures on which our prospective formula $\Phi$ must not hold.
Also, it collects a set $D$ of discarded CTL formulas that cannot be possible candidates for the prospective formula $\Phi$.
Both the sets $N$ and $D$ are designed in such a way that $\Phi$ holds on a minimal set of models, meaning that it does not hold on unnecessarily many models, as formalized in Condition 3 of Problem~\ref{prob:occ-ctl}).

We present the outline of our CEG algorithm in Algorithm~\ref{alg:ceg-algo}.
In the algorithm, we initialize our hypothesis formula $\Phi$ to be $\true$ and both the sets $N$ and $D$ to be empty.
We now exploit a procedure $\infer$ that returns a CTL formula $\Phi'$ that holds on $M$, is of size less than or equal to $B$, does not hold on any models in $N$, and is not one of the formulas in $D$.
If there is no such formula, $\infer$ reports this fact and terminates the algorithm.
We describe an implementation of the procedure $\infer$ in the next section.
Now, comparing the formulas $\Phi$ and $\Phi'$, we obtain three cases based on which we update the sets $N$ and $D$.

The first case is when $\Phi \leftrightarrow \Phi'$, that is, $\Phi$ and $\Phi'$ are equivalent, meaning they hold on the exact same set of models (Line~\ref{line:equal}). 
In this case, the CEG algorithm discards $\Phi'$, due to its equivalence to $\Phi$, adding it to $D$.
The second case is when $\Phi' \rightarrow \Phi$ and $\Phi \not\rightarrow \Phi'$, that is, $\Phi'$ holds on a proper subset of the set of models on which $\Phi$ hold (Line~\ref{line:subset}). 
In this case, the CEG algorithm generates a model $M'$ that satisfies $\Phi$ and not $\Phi$ via a CTL satisfiability subroutine~\cite{FriedmannLL10} and adds $M'$ to $N$ to eliminate $\Phi$.
Moreover, it adds $\Phi'$ to $D$ to eliminate $\Phi'$ from the search space and updates $\Phi'$ as the new hypothesis.
The final case is when $\Phi' \not\leftarrow \Phi$, that is, $\Phi'$ does not hold on a subset of the set of models on which $\Phi$ hold (Line~\ref{line:other}). 
In this case, the CEG algorithm generates a model $M'$ that satisfies $\Phi$ and not $\Phi$, again via a CTL satisfiability subroutine and adds $M'$ to $N$ to eliminate $\Phi'$.

\begin{algorithm}[t]
	\caption{CEG Algorithm for CTL formulas }\label{alg:ceg-algo}
	\textbf{Input}: Kripke structure~$M$, bound~$B$
	\begin{algorithmic}[1]
		\STATE $N\coloneqq\emptyset$, $D \coloneqq \emptyset$ 
		\STATE $\Phi\coloneqq\true$
		\WHILE{$\infer(M,B,N,D)$ returns CTL formula $\Phi'$}
		\IF{$\Phi'\leftrightarrow\Phi$}\label{line:equal}
			\STATE Add $\Phi'$ to $D$\label{line:discard1}
		\ELSE			
			\IF{$\Phi'\rightarrow\Phi$}\label{line:subset}
			 	\STATE Add $M'$ to $N$, where  $M'\models\Phi\wedge\neg\Phi'$\label{line:neg1}
			     \STATE Add $\Phi'$ to $D$\label{line:discard2}
                     \STATE $\Phi\coloneqq\Phi'$
			 \ELSE\label{line:other}
			 	\STATE Add $M'$ to $N$, where $M'\models\neg\Phi'\wedge\Phi$\label{line:neg2}
			 \ENDIF
		\ENDIF
		\ENDWHILE
		\RETURN $\Phi$
	\end{algorithmic}
\end{algorithm}

Assuming that the procedure $\infer$ performs as described (and formalized in Lemma~\ref{lem:infer-correctness}), we have the correctness of our CEG algorithm, stated as follows:
\begin{theorem}
The CEG algorithm, Algorithm~\ref{alg:ceg-algo}, terminates and always returns a solution to Problem~\ref{prob:occ-ctl}.
\end{theorem}

\begin{proof}
We first show that the algorithm terminates. 
Observe that there are only finitely many CTL formulas of size less or equal to $B$.
Moreover, $\infer$ cannot return the same formula $\Phi'$ twice; 
once returned, $\Phi'$ is eliminated from the search space either explicitly, by adding $\Phi'$ to the discard pile $D$ (Lines~\ref{line:discard1},\ref{line:discard2}), or implicitly, by adding a model $M'$ satisfying $\Phi'$ to the negative set $N$ (Line~\ref{line:neg2}).

The correctness of the algorithm relies on the correctness of the procedure $\infer$, described in Lemma~\ref{lem:infer-correctness}.
Based on that, if the procedure $\infer$ returns a formula $\Phi'$, then $\Phi'$ must hold on $M$ and must be of size less than or equal to $B$.
Thus the final solution $\Phi$ satisfies Conditions 1 and 2 of Problem~\ref{prob:occ-ctl}, since $\Phi$ is either $\true$ or a formula returned by $\infer$.

We show that the solution also satisfies Condition 3 of Problem~\ref{prob:occ-ctl} by contradiction.
In particular, we assume that our CEG algorithm returned $\Phi$, while there did exist some $\Phi'\neq\Phi$ that satisfies Conditions 1 and 2, and also, $\Phi'\rightarrow\Phi$ and $\Phi'\not\leftarrow\Phi$.
Notice that the CEG algorithm returns the final formula $\Phi$ when $\infer$ cannot return any suitable formula.
%of size less than or equal to $D$ that holds in $M$, does not hold on any models in $N$ and is not one of the formulas is $D$.
If $\infer$ did not return $\Phi'$, it could be because of two possibilities.
First, $\Phi'$ was included in $D$.
However, formulas are added to $D$ if they hold on the same or more models compared to the current hypothesis and, since $\Phi'\not\leftarrow\Phi$, $\Phi'$ should not be in $D$.
Second, $\Phi'$ hold on some model in $N$.
However, $\Phi$ does not hold in any model in $N$ and, since $\Phi'\rightarrow\Phi$, $\Phi'$ should not too.
Thus, $\infer$ would have returned $\Phi'$ if such a formula existed.
\end{proof}

\section{The Passive Learning Algorithm for CTL}\label{sec:passive-learning}

In this section, we describe the design of the procedure $\infer$ that forms a key ingredient in the CEG algorithm, Algorithm~\ref{alg:ceg-algo}.
The crux of this procedure is to solve the well-known problem of passive learning (initiated by~\cite{Gold78}) for CTL formulas.
Since passive learning is a problem of independent interest, we state the problem for CTL and present a solution to it.
Later, in this section, we describe how we utilize the solution to design the procedure $\infer$.

In passive learning, given a sample $\sample=(P,N)$ consisting of a set $P$ of positive models and a set $N$ of negative models, the goal is to find a formula $\Phi$ that is \emph{consistent} with $\sample$, that is, $\Phi$ must hold on all the models in $P$ and must not hold in any model in $N$.
Formally, we state the passive learning problem for CTL as follows:
\begin{problem}[Passive learning of CTL]\label{prob:pass-learning}
Given a sample $\sample=(P,N)$ consisting of two finite sets $P$ and $N$ of Kripke structures, find a minimal (in size) CTL formula $\Phi$ that is consistent with $\sample$.
\end{problem}
The minimality criteria for the prospective formula is to ensure that it is interpretable and does not overfit the given sample.

Our approach to solving the above problem is by reducing it to satisfiability problems in propositional logic.
We thus provide a brief introduction to propositional logic.

\paragraph{Propositional Logic.} Let $\mathit{Var}$ be a set of propositional variables that can be set to Boolean values from $\B = \{0,1\}$ (0 representing $\false$ and 1 representing $\true$).
Formulas in Propositional Logic are inductively constructed as follows:
\[
\propformula \Coloneqq x \in \mathit{Var} \mid \lnot\propformula \mid  \propformula \lor \propformula
\]

As syntactic sugar, we allow the standard Boolean formulas $\true$, $\false$, $\propformula_1 \land \propformula_2$, $\propformula_1 \rightarrow \propformula_2$, and $\propformula_1 \leftrightarrow \propformula_2$.
Note that to avoid confusion in the notation of CTL formulas and propositional formulas, we will be exclusively using $\propformula$ (along with its variants) to denote propositional formulas.

To assign values to propositional variables, we rely on a propositional valuation function $v: \mathit{Var} \to \mathbb{B}$.
We now define the semantics of propositional logic using a satisfaction relation $\models$ that is defined inductively as follows: $v \models x$ if and only if
$v(x) = 1$, $v \models \lnot\propformula$ if and only if $v \nvDash \propformula$, and $v \models \propformula_1 \lor \propformula_2$ if and only if $v \models \propformula_1$ or $v \models \propformula_2$.
In the case that $v \models \propformula$, we say that $v$ satisfies $\propformula$ and call it a \emph{satisfying assignment} of $\propformula$.
A propositional formula $\propformula$ is \emph{satisfiable} if there exists a satisfying assignment $v$ of $\Phi$.
The \emph{size} of a formula is the number of its subformulas (defined in the usual way).
The satisfiability (SAT) problem is the well-known $\NP$-complete problem of deciding whether a given propositional formula is satisfiable.
In the recent past, numerous optimized decision procedures have been designed to handle the SAT problem effectively~\cite{MouraB08,AudemardS18,BarbosaBBKLMMMN22}.

\subsection{SAT-based Learning Algorithm}

We now describe a solution to Problem~\ref{prob:pass-learning} using satisfiability problems in propositional logic, inspired by~\cite{flie,CamachoM19}.
Following their work, we design a propositional formula $\propformula^{\sample}_n$ that enables the search for a CTL formula of size $n$ that is consistent with $\sample$.
Mathematically, the formula $\propformula^{\sample}_n$ has the following properties:
\begin{enumerate}
\item $\propformula^{\sample}_n$ is satisfiable if and only if there exists a CTL formula of size $n$ that is consistent with $\sample$; and
\item from a satisfying assignment of $\propformula^{\sample}_n$, one can easily extract a suitable CTL formula of size $n$.
\end{enumerate}

Based on the formula $\propformula^{\sample}_n$, one can design an iterative algorithm to search for a minimal, suitable CTL formula: for increasing values of $n$, check the satisfiability of $\propformula^{\sample}_n$ and if satisfiable, extract a CTL formula.

Internally, the formula $\propformula^{\sample}_n$ is a conjunction of a number of subformulas with distinct roles:
\[
 \propformula^{\sample}_n \coloneqq \propformula^{\mathit{str}} \wedge \propformula^{\mathit{sem}} \wedge 
 \propformula^{\mathit{con}}
\]
The subformula $\Phi^{\mathit{str}}$ encodes the structure of the prospective CTL formula $\Phi$.
The subformula $\propformula^{\mathit{sem}}$ encodes that the correct semantics of CTL is used to interpret the prospective CTL formula on the given Kripke structures.
Finally, the subformula $\propformula^{\mathit{con}}$ ensures that the prospective CTL formula holds on the models in $P$ and not in the models in $N$.
Next, we like to describe the details of each subformula.
To keep the descriptions concise, we consider the Existence Normal Form (ENF) representation of CTL (Section 6.4 in~\cite{model-checking-book}) that consists of the propositions, Boolean operators, and the temporal modalities $\lE\lX$, $\lE\lG$ and $\lE\lU$.
It is well-known that one can always rewrite any CTL formula to its ENF representation.

\paragraph{Encoding the structure of CTL.}
To encode the structure of a CTL formula, we symbolically encode the \emph{syntax DAG} of the prospective CTL formula.
A syntax DAG is simply a syntax tree of a CTL formula in which the common nodes are merged.
Figure~\ref{fig:syntax-dag} depicts an example of a syntax tree and syntax DAG of a CTL formula.

\begin{figure}
\centering
    \begin{subfigure}[b]{0.3\textwidth}
    \centering
        \begin{tikzpicture}
                \node (1) at (0, 0) {$\lor$};
                \node (2) at (.7, -0.9) {$\lE\lU$};
                \node (3) at (-.7, -0.9) {$\lE\lX$};
                \node (4) at (0, -1.8) {$p$};
                \node (5) at (1.4, -1.8) {$\lE\lG$};
                \node (6) at (1.4, -2.7) {$q$};
                \node (7) at (-.7, -1.8) {$p$};
                \draw[->] (1) -- (2); 
                \draw[->] (1) -- (3);
                \draw[->] (2) -- (4);
                \draw[->] (2) -- (5);
                \draw[->] (3) -- (7);
                \draw[->] (5) -- (6);
        \end{tikzpicture}
        \caption{Syntax tree}
    \end{subfigure}
    \begin{subfigure}[b]{0.3\textwidth}
    \centering
    \begin{tikzpicture}
                \node (1) at (0, 0) {$\lor$};
                \node (2) at (.7, -0.9) {$\lE\lU$};
                \node (3) at (-.7, -0.9) {$\lE\lX$};
                \node (4) at (0, -1.8) {$p$};
                \node (5) at (1.4, -1.8) {$\lE\lG$};
                \node (6) at (1.4, -2.7) {$q$};
                \draw[->] (1) -- (2); 
                \draw[->] (1) -- (3);
                \draw[->] (2) -- (4);
                \draw[->] (2) -- (5);
                \draw[->] (3) -- (4);
                \draw[->] (5) -- (6);
        \end{tikzpicture}
        \caption{Syntax DAG}
    \end{subfigure}
    \begin{subfigure}[b]{0.3\textwidth}
    \centering
    \begin{tikzpicture}
                \node (1) at (0, 0) {6};
                \node (2) at (.7, -0.9) {4};
                \node (3) at (-.7, -0.9) {5};
                \node (4) at (0, -1.8) {3};
                \node (5) at (1.4, -1.8) {2};
                \node (6) at (1.4, -2.7) {1};
                \draw[->] (1) -- (2); 
                \draw[->] (1) -- (3);
                \draw[->] (2) -- (4);
                \draw[->] (2) -- (5);
                \draw[->] (3) -- (4);
                \draw[->] (5) -- (6);
        \end{tikzpicture}
        \caption{Identifiers}
    \end{subfigure}
    \caption{Representations for the CTL formula for $\lE\lX p \lor \lE(p \lU \lE\lG q )$}
    \label{fig:syntax-dag}
\end{figure}

To conveniently encode the syntax DAG of a CTL formula, we first fix a naming convention for its nodes. 
For a formula of size $n$, we assign to each of its nodes an identifier from $\{1,\dots, n\}$ such that the identifier of each node is larger than that of its children if it has any.
Note that such a naming convention may not be unique.
Based on these identifiers, we denote the subformula of $\varphi$ rooted at Node~$i$ as $\varphi[i]$.
Thus, $\varphi[n]$ denotes the entire formula $\varphi$ that is represented by the DAG.

Next, to encode a syntax DAG symbolically, we introduce the following propositional variables: (i) $x_{i,\lambda}$ for $i\in\{1,\dots, n\}$ and $\lambda\in\prop\cup\{\neg,\lor,\lE\lX,\lE\lU,\lE\lG\}$; and (ii) $l_{i,j}$ and $r_{i,j}$ for $i\in\{1,\dots,n\}$ and $j\in\{1,\dots,i\}$.
The variable $x_{i,\lambda}$ tracks the operator labeled in Node~$i$, meaning, $x_{i,\lambda}$ is set to true if and only if Node~$i$ is labeled with $\lambda$.
The variable $l_{i,j}$ (resp., $r_{i,j}$) tracks the left (resp., right) child of Node~$i$, meaning, $l_{i,j}$ (resp., $r_{i,j}$) is set to true if and only if the left (resp., right) child of Node~$i$ is Node~$j$.

We now impose structural constraints on the introduced variables to ensure they encode valid CTL formulas.
The constraints are similar to the ones proposed by Neider and Gavran~\cite{flie}.
To keep the paper self-contained, we mention them below:
\begin{align}
\Big[ \bigwedge_{1\leq i \leq n} \bigvee_{\lambda \in \Lambda} x_{i,\lambda} \Big] \land \Big[\bigwedge_{1 \leq i \leq n} \bigwedge_{\lambda \neq \lambda' \in \Lambda} \lnot x_{i,\lambda} \lor \lnot x_{i,\lambda'}  \Big],\label{eq:label-unique}\\
\Big[ \bigwedge_{2\leq i \leq n} \bigvee_{1\leq j \leq i} l_{i,j} \Big] \land \Big[\bigwedge_{2 \leq i \leq n} \bigwedge_{1 \leq j < j' < i} \lnot l_{i,j} \lor \lnot l_{i,j'}  \Big],\label{eq:left-unique}\\
\Big[ \bigwedge_{2\leq i \leq n} \bigvee_{1\leq j \leq i} r_{i,j} \Big] \land \Big[\bigwedge_{2 \leq i \leq n} \bigwedge_{1 \leq j < j' < i} \lnot r_{i,j} \lor \lnot r_{i,j'}  \Big],\label{eq:right-unique}\\
\bigvee_{p\in\prop} x_{1,p}\label{eq:only-prop}
\end{align}
where $\Lambda=\{\neg,\lor,\lE\lX,\lE\lU,\lE\lG\}$.
The constraint~\ref{eq:label-unique} ensures that each node is labeled by exactly one operator or one proposition.
The constraints~\ref{eq:left-unique} and~\ref{eq:right-unique} enforce that each node (except for Node 1) has exactly one left and exactly one right child, respectively.
Finally, the constraint~\ref{eq:only-prop} ensures that Node~$1$ is labeled by a proposition.
$\Phi^{str}$ is simply the conjunction of all the structural constraints~\ref{eq:label-unique} to~\ref{eq:only-prop}.

Based on a satisfying assignment $v$ of $\Phi^{\mathit{str}}$, one can easily construct a unique CTL formula as follows: label each Node~$i$ with the unique operator $\lambda$ for which $v(x_{i,\lambda})=1$, and mark the left and right children of the node with the unique Nodes $j$ and $j'$ for which $v(l_{i,j})=1$ and $v(r_{i,j'})=1$, respectively.

\paragraph{Encoding the semantics of CTL.}
In order to symbolically encode the semantics of the prospective CTL formula $\Phi$ for a given Kripke structure $M$, we rely on encoding a well-known CTL model-checking procedure (Section 6.4 in~\cite{model-checking-book}).
The procedure involves calculating, for each subformula $\Phi'$ of $\Phi$, the set $\Sat(\Phi') = \{s\in S~|~ s\models \Phi\}$ which consists of the states of $M$ where $\Phi'$ holds.
The computation of the $\Sat(\Phi')$ is done recursively on the structure of $\Phi$ based on fixed-point computations.
For a given Kripke structure $M$, the computations of $\Sat(\Phi')$ for the different possible CTL formulas in the ENF representation are stated below.
\begin{align}
&\Sat(p) = \{ s \in S~|~p \in L(s) \}, \text{ for any }p\in \prop,\\
&\Sat(\Phi \wedge \Psi) = \Sat(\Phi) \cap \Sat(\Psi),\\
&\Sat(\neg\Phi) = S \setminus \Sat(\Phi),\\
&\Sat(\lE\lX\Phi) = \{ s \in S~|~\post(s) \cap \Sat(\Phi) \neq \emptyset \},\\
&\Sat(\lE(\Phi \lU \Psi)) \text{ is the smallest subset }T \text{ of }S, \text{ such that }\nonumber\\
&\hspace{1cm}(1)\ \Sat(\Psi) \subseteq T\text{ and }
(2)\ s \in \Sat(\Phi)\text{ and }\post(s) \cap T \neq \emptyset \text{ implies }s \in T,\label{eq:sat-EU}\\
&\Sat(\lE\lG\Phi)\text{ is the largest subset }T \text{ of }S,\text{ such that }\nonumber\\
&\hspace{1cm}(1)\ T \subseteq \Sat(\Phi)\text{ and }(2)\ s \in T\text{ implies }\post(s) \cap T \neq \emptyset\label{eq:sat-EG}
\end{align}

We intend to symbolically encode the above-described $\Sat$ set computation of CTL formulas.
To do so, intuitively, our approach is to introduce propositional variables that can encode whether a state $s$ of Kripke structure $M$ belongs to $\Sat(\Phi[i])$ for all subformulas $\Phi[i]$ of $\Phi$.
Formally, we introduce propositional variables $y^{M}_{i,s}$ for each $i\in\{1,\dots,n\}$, $s\in \{1,\ldots,|S|\}$ and $M\in P\cup N$.
The variable $y^{M}_{i,s}$ tracks whether subformula $\Phi[i]$
of $\Phi$ holds in state $s$, meaning, $y^{M}_{i,s}$ is set to true if and only if $\Phi[i]$ holds in state $s$ of $M$, i.e., $s\in\Sat(\Phi[i])$.

To ensure the desired meaning of the variables, we impose the following constraints:
\begin{align}
\bigwedge_{1\leq i\leq n} x_{i,p} \rightarrow \bigwedge_{s\in S} \big[ y^{M}_{i,s} \leftrightarrow p\in L(s)\big]\label{eq:prop}\\
\bigwedge_{\substack{{1\leq n}\\{1\leq j,j' < i}}} [x_{i,\wedge}\wedge l_{i,j}\wedge r_{i,j'}] \rightarrow \bigwedge_{s\in S}\big[ y^{M}_{i,s} \leftrightarrow [y^{M}_{j,s} \wedge y^{M}_{j',s}]\big]\label{eq:and}\\
\bigwedge_{\substack{{1\leq n}\\{1\leq j < i}}} [x_{i,\neg}\wedge l_{i,j}] \rightarrow \bigwedge_{s\in S}\big[ y^{M}_{i,s} \leftrightarrow \neg y^{M}_{j,s}\big]\label{eq:not}\\
\bigwedge_{\substack{{1\leq n}\\{1\leq j < i}}} [x_{i,\lE\lX}\wedge l_{i,j}] \rightarrow \bigwedge_{s\in S}\big[ y^{M}_{i,s} \leftrightarrow \bigvee_{s'\in \post(s)} y^{M}_{j,s'}\big]\label{eq:EX}
\end{align}
The above constraints encode the $\Sat$ set computation for the propositions, Boolean operators and the $\lE\lX$ operator.

While the constraints for the propositions, Boolean operators and the $\lE\lX$ operators are a straightforward translation of the $\Sat$ set computations, the constraints for the $\lE\lU$ and $\lE\lG$ operators require some innovation.
This is because they require the computation of the least and greatest fixed-points, respectively, as can be seen from the equations~\ref{eq:sat-EU} and~\ref{eq:sat-EG}.

For the operators $\lE\lU$ and $\lE\lG$, we mimic the steps of the fixed-point computation algorithm (Algorithm 15 and 16 from~\cite{model-checking-book}, respectively) in propositional logic.
For both operators $\lE\lU$ and $\lE\lG$, the fixed-point computation algorithm internally maintains an estimate of the $\Sat$ set and updates it in iteratively.
Thus, to encode the fixed-point computation, we now introduce propositional variables that can encode whether a state $s$ of Kripke structure $M$ belongs to a particular estimate of $\Sat$.

Formally, we introduce propositional variables $y^{M}_{i,s,k}$ for each $i\in\{1,\ldots,n\}$,  $s\in \{1,\ldots,|S|\}$, $k\in\{1,\ldots,|S|+1\}$, and $M\in P\cup N $ that are responsible for encoding estimates of $\Sat$.
The parameter $k$ in the variable $y^{M}_{i,s,k}$ tracks which iterative step of the fixed-point computation the variable encodes.
Also, as the fixed-point algorithms need at most $|S|+1$ iterative steps, $k$ ranges from $1$ to $|S|+1$.

We now impose the following constraints on the introduced variables to ensure their desired meaning:
\begin{align}
\bigwedge_{\substack{{1\leq i\leq n}\\{1\leq j,j' < i}}} [x_{i,\lE\lU}\wedge l_{i,j}\wedge r_{i,j'}] \rightarrow \nonumber\\
\bigwedge_{s\in S}\Big[
\big[y^{M}_{i,s,1} \leftrightarrow y^{M}_{j',s}\big] \wedge
\bigwedge_{1\leq k\leq |S|}
\big[
y^{M}_{i,s,k+1} \leftrightarrow [y^{M}_{i,s,k} \vee [y^{M}_{j,s} \wedge \bigvee_{s'\in\post(s)} y^{M}_{i,s',k}]]
\big] \wedge [y^{M}_{i,s} \leftrightarrow y^{M}_{i,s,|S|+1}]
\Big]\label{eq:EU}\\
\bigwedge_{\substack{{1\leq i\leq n}\\{1\leq j < i}}} [x_{i,\lE\lG}\wedge l_{i,j}] \rightarrow \nonumber\\
\bigwedge_{s\in S}\Big[
\big[y^{M}_{i,s,1} \leftrightarrow y^{M}_{j,s}\big] \wedge
\bigwedge_{1\leq k\leq |S|}
\big[
y^{M}_{i,s,k+1} \leftrightarrow [y^{M}_{j,s} \wedge \bigvee_{s'\in\post(s)} y^{M}_{i,s',k}]
\big]\Big] \wedge 
[y^{M}_{i,s} \leftrightarrow y^{M}_{i,s,|S|+1}]\label{eq:EG}
\end{align}

To show that our encoding indeed captures the meaning of the variables $y^{M}_{i,s,k}$, we present the following lemmas.
\begin{lemma}\label{lem:EU-correctness}
    Let $v$ be an assignment such that $v(x_{i,\lE\lU})=1$, $v(l_{i,j})=1$ and $v(r_{i,j'})=1$.
    Then, for any state $s\in S$ in $M$ and $k\in \{1,\dots,|S|+1\}$, we have: $v(y^{M}_{i,s,k})=1$ if and only if there exists a path prefix $\pi = s_0 s_1 s_2\ldots s_t$ in $M$ where $s=s_0$ and $0\leq t < k$, such that $v(y^{M}_{j,\overline{s}})=1$ for all $\overline{s} \in \{s_0, s_1, \ldots, s_{t-1}\}$ and $v(y^{M}_{j',s_t})=1$.
\end{lemma}
\begin{proof}
The proof requires to proceed via induction on the parameter $k$.
    \begin{itemize}
    \item As the base case, let $k=1$. 
    Using the constraint~\ref{eq:EU}, we have $v(y^{M}_{i,s,1})=v(y^{M}_{j',s})$. 
    Thus, $v(y^{M}_{i,s,1})=1$ if and only if there is a trivial path $\pi=s$ (with parameter $t<1$) such that $v(y^{M}_{j',s})=1$.

    \item As the induction hypothesis, let for any $s\in S$ and $k\leq K$ the following hold: $v(y^{M}_{i,s,K})=1$ if and only if there is a path prefix $\pi=s s_1 s_2 \ldots s_t$, $t< K$ such that $v(y^{M}_{j,\overline{s}})=1$ for all $\overline{s} \in \{s, s_1, \ldots, s_{t-1}\}$ and $v(y^{M}_{j',s_t})=1$.
    Now, using the constraint~\ref{eq:EU}, $v(y^{M}_{i,s,K+1})=1$ if and only if (i) $v(y^{M}_{i,s,K})=1$, meaning, there is a path prefix  $\pi=s s_1 s_2 \ldots s_t$, $t<K<K+1$ such that $v(y^{M}_{j,\overline{s}})=1$ for all $\overline{s} \in \{s, s_1, \ldots, s_{t-1}\}$ and $v(y^{M}_{j',s_t})=1$; or (ii) $v(y^{M}_{j',s})=1$ and for some $s'\in\post(s)$, $v(y^{M}_{i,s',K})=1$, meaning, there is a path prefix $\pi=s s' \ldots s_t$, $t<K+1$ such that $v(y^{M}_{j,s'})=1$ for all $\overline{s} \in \{s, s_1, \ldots, s_{t-1}\}$ and $v(y^{M}_{j',s_t})=1$.
    \end{itemize}
\end{proof}

\begin{lemma}\label{lem:EG-correctness}
    Let $v$ be an assignment such that $v(x_{i,\lE\lG})=1$ and $v(l_{i,j})=1$.
    Then, for any state $s\in S$ in $M$ and $k\in \{1,\dots,|S|+1\}$, we have: $v(y^{M}_{i,s,k})=1$ if and only if there exists a path prefix $\pi = s_0 s_1 s_2\ldots s_t$ in $M$ where $s=s_0$ and $t \geq k-1$, such that $v(y^{M}_{j,\overline{s}})=1$ for all $\overline{s} \in \{s_0, s_1, \ldots, s_{t}\}$.
\end{lemma}
\begin{proof}
The proof requires an induction on the parameter $k$ similar to the previous proof.
    \begin{itemize}
    \item As the base case, let $k=1$. 
    Using the constraint~\ref{eq:EG}, we have $v(y^{M}_{i,s,1})=v(y^{M}_{j,s})$. 
    Thus, $v(y^{M}_{i,s,1})=1$ if and only if there is a trivial path $\pi=s$ (with parameter $t\geq 0$) such that $v(y^{M}_{j,s})=1$.

    \item As the induction hypothesis, let for any $s\in S$ and $k\leq K$ the following hold: $v(y^{M}_{i,s,K})=1$ if and only if there is a path prefix $\pi=s s_1 s_2 \ldots s_t$, $t\geq K-1$ such that $v(y^{M}_{j,\overline{s}})=1$ for all $\overline{s} \in \{s, s_1, \ldots, s_{t}\}$.
    Now, using the constraint~\ref{eq:EG}, $v(y^{M}_{i,s,K+1})=1$ if and only if  $v(y^{M}_{j,s})=1$ and for some $s'\in\post(s)$, $v(y^{M}_{i,s',K})=1$, meaning, there is a path prefix $\pi=s s' \ldots s_t$, $t\geq K$ such that $v(y^{M}_{j,s'})=1$ for all $\overline{s} \in \{s, s_1, \ldots, s_{t}\}$.
    \end{itemize}
\end{proof}

The formula $\propformula^{\mathit{sem}}$ now is simply the conjunction of all the semantic constraints,~\ref{eq:prop} to~\ref{eq:EG}.

\paragraph{Encoding the consistency with the models}
Finally, to encode that the prospective formula is consistent with $\sample$, we have the following constraint:
\[
\big[\bigwedge_{M\in P}\bigwedge_{s\in I} y^{M}_{n,s}\big] \wedge \big[\bigwedge_{M\in N}\bigvee_{s\in I} \neg y^{M}_{n,s}\big],
\]
which encodes that the prospective formula must hold in all of the initial states for the models in $P$ and must not hold in some initial state for the models in $N$.

We now state the following lemma to state the correctness of the encoding 
\begin{lemma}\label{lem:pass-learning-correctness}
Let $\sample=(P,N)$ be a sample, $n \in \mathbb{N}\setminus \{0\}$
and $\propformula^{\sample}_n$ the propositional formula described above.
Then, the following holds:
\begin{enumerate}
\item If a CTL formula of size $n$ consistent with $\sample$ exists, then the propositional formula $\propformula^{\sample}_n$ is satisfiable.
\item If $v\models\propformula^{\sample}_n$, then $\Phi^v$ is a CTL formula of size $n$ that is consistent with $\sample$.
\end{enumerate}
\end{lemma}

\begin{proof}
The proof of the first part of the above lemma proceeds almost identically to the proof of Lemma 1 in~\cite{flie}.
From a prospective CTL formula, say $\Phi$, first one constructs an assignment $v$.
This is done by setting the values $v(x_{i,\lambda})$, $v(l_{i,j})$, $v(r_{i,j})$, $v(y^{M}_{i,s})$ and $v(y^{M}_{i,s,k})$ based on the structure and semantics of $\Phi$.
For instance, one sets $v(x_{i,\lambda})=1$ if and only if Node~$i$ of the syntax DAG of $\Phi$ is labeled with operator $\lambda$, one sets $v(l_{i,j})$ and $v(r_{i,j})$ according to the structure of the syntax DAG of $\Phi$ and so on.
One can check that such an assignment $v$ satisfies the structural constraints,~\ref{eq:label-unique} to~\ref{eq:only-prop}.
A similar check can also be done for the semantic constraints,~\ref{eq:prop} to~\ref{eq:EG}.

For proving the second part, we consider $v \models \propformula^{\sample}_n$ and $\Phi^v$ to be the CTL formula obtained from $v$.
In particular, $v \models \Phi^\mathit{str}$, and, thus, the variables $x_{i,\lambda}$, $l_{i,j}$, and $r_{i,j}$ encode the syntax DAG of a unique formula $\Phi^v$.
We now show by induction on the structure of $\Phi^{v}$ that $v(y^{M}_{i,s}) = 1$ if and only if $s\in \Sat(\Phi[i])$ for each subformula $\Phi[i]$ of $\Phi^v$ and each model $M\in P\cup N$:
\begin{itemize}[itemsep=10pt]
    \item Let $\Phi[i] = p$ for a proposition $p \in \mathcal P$, that is, $v(x_{i, p}) = 1$.
    Then, the constraint~\ref{eq:prop} immediately ensures that $v(y^{M}_{i, s}) = 1$ if and only if $s\in\Sat(p)$ for the model $M$.

    \item Let $\Phi[i] = \Phi[j] \wedge \Phi[j']$, that is, $v(x_{i, \wedge}) = 1$.
    Thus, combining the induction hypotheses together with the constraint~\ref{eq:and}, we have that $v(y^{M}_{i, s}) = 1$ if and only if $s\in\Sat(\Phi[i])$ and $s\in\Sat(\Phi[j'])$.
    
    \item Let $\Phi[i] = \neg \Phi[j]$, that is, $v(x_{i, \lnot}) = 1$.
    Thus, combining the induction hypothesis together with the constraint~\ref{eq:not}, we have $v(y^{M}_{i, s}) = 1$ if and only if $s\not\in\Sat(\Phi[j])$.
    
    \item Let $\Phi[i] = \lE\lX \Phi[j]$, that is, $v(x_{i, \lE\lX}) = 1$.
    Thus, combining the induction hypothesis together with the constraint~\ref{eq:EX}, we have $v(y^{M}_{i, s}) = 1$ if and only if there is some $s'\in\post(s)$ such that $s'\in\Sat(\Phi[j]))$.

    \item Let $\Phi[i] = \lE(\Phi[j]\lU\Phi[j'])$, that is, $v(x_{i, \lE\lU}) = 1$. 
    Using the constraint~\ref{eq:EU}, we have $v(y^{M}_{i,s})=v(y^{M}_{i,s,|S|+1})$. 
    Thus, combining the induction hypothesis with Lemma~\ref{lem:EU-correctness}, we have: $v(y^{M}_{i,s})$ if and only there exists a path prefix $\pi = s s_1 s_2\ldots s_t$, $0\leq t \leq |S|$ such that $\overline{s}\in\Sat(\Phi[j])$ for all $\overline{s} \in \{s_0, s_1, \ldots, s_{t-1}\}$ and $s_t\in\Sat(\Phi[j'])$.
    This is a necessary and sufficient condition that leads to $s\in\Sat(\Phi[i])$.
    
    \item Let $\Phi[i] = \lE\lG(\Phi[j])$, that is, $v(x_{i, \lE\lG}) = 1$. 
    Using the constraint~\ref{eq:EG}, we have $v(y^{M}_{i,s})=v(y^{M}_{i,s,|S|+1})$. 
    Thus, combining the induction hypothesis with Lemma~\ref{lem:EU-correctness}, we have: $v(y^{M}_{i,s})$ if and only if there exists a path prefix $\pi = s s_1 s_2\ldots s_t$, $t\geq |S|$ such that $\overline{s}\in\Sat(\Phi[j])$ for all $\overline{s} \in \{s_0, s_1, \ldots, s_{t}\}$.
    This is a necessary and sufficient condition that leads to $s\in\Sat(\Phi[i])$.
 \end{itemize}

\end{proof}

\subsection{Implementation of Infer}

For constructing the procedure $\infer$, we reuse the solution devised for the passive learning problem, Problem~\ref{prob:pass-learning}.
Here, we additionally require the prospective CTL formula to be not one of the discarded formulas present in the set $D$.
As a result, we design a propositional formula $\propformula^{\neq \Phi}$ that ensures that the prospective CTL formula is not syntactically equal to $\Phi$.
We can extend this formula to design the following formula:
\[
\propformula^D\coloneqq \bigwedge_{\Phi\in D} \propformula^{\neq \Phi}.
\]
We incorporate this formula in our solution to the passive learning problem to obtain the procedure $\infer$, which is sketched in Algorithm~\ref{alg:infer}.

\begin{algorithm}[t]
\caption{Procedure $\infer$}\label{alg:infer}
\textbf{Input}: Parameters $M$, $B$, $N$, $D$
\begin{algorithmic}[1]
    \STATE $\sample \coloneqq (P,N)$, where $P=\{M\}$
    \FOR{$n=1,2,\dots,B$}
    \STATE Construct formula $\propformula^{\sample}_n\wedge \propformula^{D}$
    \IF{$\propformula^{\sample}_n\wedge \propformula^{D}$ is satisfiable (with a satisfying assignment $v$)}
        \RETURN SAT, $\Phi^v$
    \ENDIF
    \ENDFOR
    \RETURN No solution exists
\end{algorithmic}
\end{algorithm}

We finally state the correctness of the procedure $\infer$ in the following theorem.
The proof of the theorem follows directly from the correctness (Lemma~\ref{lem:pass-learning-correctness}) of encoding for the passive learning algorithm.
\begin{theorem}\label{lem:infer-correctness}
Given parameters $M$, $B$, $N$ and $D$ (mentioned in Section~\ref{sec:ceg-algorithm}), the procedure \text{Infer} either returns SAT along with a CTL formula $\Phi$ that holds on $M$, is of size less than or equal to $B$, does not hold in any Kripke structures in $N$ and is not one of the formulas in $D$, if such a formula exists; or returns UNSAT.
\end{theorem}

\section{Conclusion}
We study the problem of inferring a concise specification in CTL that suitably describes the behavior of a given Kripke structure.
To infer such a specification, we rely on size and language minimality of CTL formulas as regularizers.
We design a counter-example guided (CEG) algorithm that, in an iterative manner, infers a concise and language-minimal CTL formula.
As a subroutine, the CEG algorithm relies on solving the passive learning problem for CTL, which is to learn a minimal CTL formula consistent with a given set of desirable and undesirable Kripke structures.
For the passive learning problem, we devise a solution based on iterative SAT solving.
We prove that all our algorithms are sound and complete, meaning that they always output an optimal CTL formula.

As future work, we plan to implement the inference algorithms in an open-source tool and study their ability to infer interpretable formulas. 
Further, we plan to investigate the task of inferring formulas in other branching-time logics, such as CTL*~\cite{CTLstar} (which is more expressive than CTL), and ATL~\cite{AlurATL} (which can express properties over concurrent structures).

\subsubsection*{Acknowledgments}
The work was partly funded by the Deutsche Forschungsgemeinschaft (DFG, German Research Foundation) grant number 434592664.

\bibliographystyle{splncs04}
\bibliography{refs}

\end{document}